\documentclass[
 reprint, superscriptaddress,
 amsmath,amssymb,
 aps,
 prl,
 a4paper]{revtex4-2}
\pdfoutput=1 %

\usepackage{braket}
\usepackage{amsmath,amssymb,amsfonts}
\usepackage{amsthm}
\usepackage{mathtools}
\usepackage{graphicx}
\usepackage{xcolor}
\usepackage{framed}
\usepackage[colorlinks=true,
            allcolors=blue]{hyperref}
\usepackage[capitalize,nameinlink]{cleveref}
\usepackage{comment}

\newcommand{\be}{\begin{equation}}\newcommand{\ee}{\end{equation}}\newcommand{\ba}{\begin{eqnarray}}\newcommand{\ea}{\end{eqnarray}}\newcommand{\ban}{\begin{eqnarray*}}\newcommand{\ean}{\end{eqnarray*}}

\newcommand{\one}{\mathbb{I}}

\newtheorem{theorem}{Theorem} %

\newtheorem{corollary}{Corollary}[theorem]

\newtheorem{definition}{Definition}
\theoremstyle{definition}

\theoremstyle{remark}

\def\sH{\mathcal{H}}

\def\openone{\mathds{1}}

\newcommand{\X}{\mathcal{X}}

\renewcommand{\set}[1]{\mathcal{#1}}

\usepackage{dsfont}
\usepackage{enumitem}

\newcommand{\Tr}{\mathrm{Tr}}

\newcommand{\tN}{\mathbf{N}}
\newcommand{\tU}{\mathbf{U}}
\newcommand{\tA}{\mathbf{A}}

\renewcommand{\ge}{\geqslant}

\renewcommand{\geq}{\geqslant}
\renewcommand{\leq}{\leqslant}

\newcommand{\IWY}{I_{\rm WY}}

\DeclareMathOperator{\diag}{diag}

\begin{document}

\title{Quantifying Measurement Objectivity: A Retrodictive Approach}

\author{Jiaxi Kuang}
\email{kuang.jiaxi.g5@s.mail.nagoya-u.ac.jp}
\affiliation{Department of Mathematical Informatics, Nagoya University, Furo-cho Chikusa-ku, Nagoya 464-8601, Japan}

\author{Teruaki Nagasawa}
\email{teruaki-nagasawa@se.kanazawa-u.ac.jp}
\affiliation{Department of Mathematical Informatics, Nagoya University, Furo-cho Chikusa-ku, Nagoya 464-8601, Japan}
\affiliation{Institute of Science and Engineering, Kanazawa University, Kanazawa, Ishikawa, 920-1192, Japan}

\author{Kensei Torii}
\email{toriikensei@nagoya-u.jp}
\affiliation{Department of Mathematical Informatics, Nagoya University, Furo-cho Chikusa-ku, Nagoya 464-8601, Japan}

\author{Francesco Buscemi}
\email{buscemi@nagoya-u.jp}
\affiliation{Department of Mathematical Informatics, Nagoya University, Furo-cho Chikusa-ku, Nagoya 464-8601, Japan}

\begin{abstract}
When can one interpret the outcomes of a quantum measurement as revealing a pre-existing objective property? Using the recently developed formalism of quantum measurement retrodiction, we provide a quantitative treatment of this question: for any POVM and faithful prior state, we construct a positive semidefinite bilinear form that quantifies the non-objectivity of every real-valued outcome feature through the disagreement between its predictive value and its retrodictive counterpart. We show that this form decomposes exactly into the sum of two positive semidefinite bilinear forms: an unsharpness form and an asymmetry form given by Wigner--Yanase skew information. The total form vanishes precisely on those outcome features that can be interpreted, relative to the prior, as revealing pre-existing properties; in particular, it vanishes identically if and only if the POVM is sharp and commutes with the prior. Finally, under maps that preserve the prior and are covariant under its modular group, asymmetry cannot increase, and any loss of asymmetry is offset by at least as much unsharpness, so that total non-objectivity cannot decrease.
\end{abstract}
\maketitle

\noindent {\em Introduction.}---In one possible sense, a measurement outcome may be regarded as \textit{objective} when it admits an interpretation as a faithful record of a property already present \textit{before} the measurement. This property-revealing idea is related to the EPR discussion of physical reality~\cite{einstein1935can}. However, the Bell and Kochen--Specker theorems place well-known restrictions on property-revealing interpretations of quantum measurements~\cite{Bell1964,KochenSpecker1967,mermin1993hidden}. In a classical theory, by contrast, such an interpretation is straightforward whenever the outcome is a function of the system's pre-existing state. Here we ask the corresponding, narrower question: given a positive operator-valued measure (POVM) and a prior state, when does the measurement admit this classical property-revealing interpretation? If the complete outcome does not, which functions of the outcome, if any, retain it, and how can the remaining non-objectivity be quantified?

This prior-relative question differs from the emergence of objectivity studied in quantum Darwinism, strong quantum Darwinism, and spectrum broadcast structures, which concerns how information about a system becomes redundantly encoded in its environment and independently accessible to many observers~\cite{ollivier2004objective,blume2006quantum,zurek2009quantum,le2019strong,korbicz2014objectivity,horodecki2015quantum}. It also differs from \emph{intersubjectivity}, which concerns agreement among different observers on the outcomes of their measurements~\cite{chisholm2023redundancy,poderini2023nonobjectivity,ozawa2025intersubjectivity,candeloro2026intersubjectivity,umekawa2026intersubjectivity}.

Here we argue that the formalism of quantum measurement retrodiction provides an operational criterion for this question. Given an observed outcome, a \textit{principle of minimum change} provides a state assignment for retrodictive inferences about the system \textit{before} the measurement~\cite{bai-2025-q-bayes-minimum-change,kuang2026quantum}. Evaluating the same POVM on this state tests whether it predicts the observed outcome with certainty. We show that prediction and retrodiction agree in this sense for every outcome exactly when the POVM is sharp and commutes with the prior. This criterion is distinct from \textit{post-measurement repeatability}: repeatability depends on the instrument implementing the POVM, whereas retrodictive agreement depends only on the POVM and the prior.

Even when the complete outcome fails this criterion, a coarse-graining of the outcome may still reveal a sharp property compatible with the prior. We first identify the finest such coarse-graining~\cite{Nagasawa-2025-macrostates-ROPP}; its values partition the outcomes into blocks. Then, within each resulting block, we consider arbitrary real-valued functions of the outcome, which we call \emph{outcome features}. The Hirschfeld--Gebelein--R\'enyi maximum correlation quantifies the greatest retrodictive agreement attainable by a normalized feature~\cite{Hirschfeld_1935,gebelein1941statistische,renyi1959measures}. However, rather than focusing only on this scalar quantity, our main result describes non-objectivity through three positive-semidefinite bilinear forms, or \textit{tensors}, resolving it in every ``feature direction''. The tensor of total non-objectivity splits exactly into two contributions, an unsharpness tensor and an asymmetry tensor. Unsharpness measures the failure of the outcomes to define mutually exclusive sharp properties and may remain nonzero even when the measurement commutes with the prior~\cite{hall2004prior,busch2004noise,massar2007uncertainty,liu2021quantifying}. Asymmetry is given by Wigner--Yanase skew information and measures noncommutativity with the prior~\cite{wigner1963information}. When optimized over normalized features, the total tensor gives back the scalar HGR non-objectivity. The two contributions vanish together exactly when the POVM is a projection-valued measure commuting with the prior, precisely the objective case identified above.

Finally, we establish the monotonicity of all three tensors, in the Loewner order, under a physically relevant class of maps that we call $\gamma$-covariant. These are maps that commute with the modular group of $\gamma$ and whose duals preserve $\gamma$. Such maps cannot increase asymmetry, while any decrease in asymmetry is accompanied by an increase in unsharpness that is at least as large. Total non-objectivity therefore cannot decrease. This single tensor inequality controls every outcome feature, as well as the optimized HGR non-objectivity.

\ 

\noindent {\em Retrodictive criterion for objectivity.}---Let $\mathsf P=\{P_x\}_{x\in\X}$ be a finite POVM on $\sH$ and let $\gamma>0$ be a faithful prior. According to Born's rule, we write $p_x:=\Tr[\gamma P_x]$. We assume that the POVM contains no null effects, so that $p_x>0$ for every $x$.

Classically, Bayes' rule follows from the minimum-change principle (MCP): upon acquiring new data, one updates the reverse conditional distribution so as to incorporate them while altering the prior joint input--output process as little as possible. Recent works generalized MCP to the quantum setting~\cite{bai-2025-q-bayes-minimum-change,kuang2026quantum}. For a POVM $\mathsf P$ and a faithful prior $\gamma$, they showed that, upon obtaining outcome $x$, MCP uniquely yields the retrodictive update
\begin{equation}
    \gamma_x:=\frac{\sqrt\gamma P_x\sqrt\gamma}{p_x}.
\end{equation}
Thus, $\gamma_x$ is the quantum analogue of the Bayesian posterior: it provides a state assignment, conditioned on the observed outcome, for retrodictive inferences about the system before the measurement.

Using it to predict the outcome of the same POVM gives the retrodictive joint law
\begin{equation}\label{eq:joint-law}
    T_{xx'}:=p(x,x')=\Tr[\sqrt\gamma P_x\sqrt\gamma P_{x'}].
\end{equation}
It is nonnegative and symmetric. It is also positive semidefinite, since it is the Gram matrix of the operators $\gamma^{1/4}P_x\gamma^{1/4}$ in the Hilbert--Schmidt inner product. Its marginals are
\begin{equation*}\label{eq:marginal}
    \sum_{x'}T_{xx'}=p_x,
    \qquad
    \sum_xT_{xx'}=p_{x'},
\end{equation*}
and $\sum_{xx'}T_{xx'}=1$.

\begin{definition}[Retrodictive criterion]\label{def:measurement_objectivity}
The retrodictive criterion for a POVM $\mathsf P$ to be objective relative to $\gamma$ is that its retrodictive joint law be diagonal, i.e.,
\begin{equation}
    T_{xx'}=p_x\delta_{xx'}
    \qquad\forall x,x'.
\end{equation}
\end{definition}

Operationally, the retrodictive criterion says that prediction and retrodiction agree. Since $p_x>0$, it is equivalent to $\Tr[\gamma_xP_{x'}]=\delta_{xx'}$: the retrodictive state predicts the observed outcome with certainty. This is not a statement about the repeatability of an instrument, because $\gamma_x$ is not a post-measurement output state.

As we will see in what follows, the physical content of the retrodictive criterion comes from Theorem~\ref{thm:tensor-zero-conditions}: it is satisfied exactly when $\mathsf P$ is a PVM and $[\gamma,P_x]=0$ for every $x$. Under these conditions,
\begin{align}
    \gamma_x&=\frac{P_x\gamma P_x}{p_x},
    &\gamma&=\sum_xp_x\gamma_x,\notag\\
    P_x\gamma_x&=\gamma_x,
    &\Tr[\gamma_xP_{x'}]&=\delta_{xx'}.
    \label{eq:objective-sector-decomposition}
\end{align}
Thus the prior admits a canonical decomposition, relative to $\mathsf P$, into states supported on mutually exclusive sharp sectors. In the component $\gamma_x$, the property $P_x$ is definite. Moreover, every real-valued function $v=(v_x)\in\mathbb{R}^{\set{X}}$ defines the operator
\begin{equation}\label{eq:feature-observable}
    O_v:=\sum_xv_xP_x.
\end{equation}
Under the present conditions, $O_v$ is a sharp observable satisfying $O_v\gamma_x=v_x\gamma_x$. Hence every function of the outcome has a definite value in each sector, and the measurement is objective relative to $\gamma$ in the sense that it can be regarded as revealing which pre-existing value obtains rather than creating that value.

Note that Eq.~\eqref{eq:objective-sector-decomposition} is relative to the observable-prior pair and does not determine the preparation history or measurement implementation. If the decomposition is a proper mixture, or \emph{Gemenge}, the weights $p_x$ have a literal classical-ignorance interpretation; for an improper reduced state, no such inference about individual realizations follows~\cite{mittelstaedt1998interpretation,despagnat1976conceptual}. Likewise, a POVM does not fix an instrument~\cite{davies1970operational}: although a commuting PVM admits a repeatable L\"uders instrument whose nonselective action leaves $\gamma$ invariant~\cite{luders2006statechange}, the same PVM may be implemented destructively.

\ 

\noindent {\em Exact objective information.}---The retrodictive criterion in Definition~\ref{def:measurement_objectivity} asks whether the complete outcome label can be interpreted as revealing a pre-existing sharp property. This is a strong requirement. Its failure does not imply that no part of the outcome is objective relative to $\gamma$: a nontrivial deterministic function of the outcome may still label mutually exclusive sharp sectors of the prior.

The mere existence of such a function is not informative by itself. Indeed, the constant function maps every outcome to the same value and produces the one-outcome POVM $\{\openone\}$, which is always objective relative to $\gamma$ but carries no information. We must therefore identify the \emph{finest} objective deterministic coarse-graining and then quantify the distinctions that remain.

A deterministic coarse-graining $f:\X\to\mathcal Y$ produces the POVM $\Pi_y^{(f)}:=\sum_{x:f(x)=y}P_x$. By Definition~\ref{def:measurement_objectivity} and Theorem~\ref{thm:tensor-zero-conditions}, its output is objective relative to $\gamma$ precisely when $\{\Pi_y^{(f)}\}_y$ is a PVM commuting with $\gamma$. In that case, $y$ labels a pre-existing sharp sector determined by the original outcome $x$. We call such a coarse-graining a prior-compatible projective post-processing.

Previous work~\cite{Nagasawa-2025-macrostates-ROPP} shows that every pair $(\mathsf P,\gamma)$ admits, up to relabeling, a unique \textit{finest} prior-compatible projective post-processing. We denote its PVM by $\Pi=\{\Pi_y\}_{y\in\mathcal Y}$ and call it the maximal prior-compatible projective post-processing (MPPP). Its fibers $\X_y:=f^{-1}(y)$ form a disjoint partition of $\X$, with $\Pi_y=\sum_{x\in\X_y}P_x$. Thus, $y$ is the finest pre-existing sharp property that can be inferred deterministically from $x$.

We may therefore analyze separately what remains within each block $\X_y$. On the subspace $\mathcal H_y:=\Pi_y\mathcal H$, the corresponding prior and POVM are
\begin{equation}
    \gamma^{(y)}:=\frac{\Pi_y\gamma\Pi_y}{\Tr[\gamma\Pi_y]},
    \qquad
    \mathsf P^{(y)}:=\{P_x|_{\mathcal H_y}\}_{x\in\X_y}.
\end{equation}
A one-outcome block is trivial and requires no further analysis. In every nontrivial block, maximality of the MPPP rules out any further sharp property compatible with the prior. In what follows, we fix one such block. All quantities appearing there, including $p_x$, $T$, and $D$, are constructed from the restricted pair $(\mathsf P^{(y)},\gamma^{(y)})$, but, for the sake of readability, we suppress the superscript $y$.

\ 

\noindent {\em Feature extraction.}---Within the fixed nontrivial MPPP block, the sharp sector label revealed by the outcome has already been extracted. We now quantify the objectivity of the remaining distinctions. Recall the real-valued functions $v=(v_x)\in\mathbb{R}^{\set{X}}$ introduced in Eq.~\eqref{eq:feature-observable}. We call each such function an \emph{outcome feature}. Let us set
\begin{equation}
    D:=\diag(p_x),\qquad \tN:=D-T.
\end{equation}
Let $X$ and $X'$ be random variables on $\set{X}$ with joint law Eq.~\eqref{eq:joint-law}. Here, $X$ is the ``predictive'' outcome label and $X'$ its ``retrodictive'' counterpart: conditional on $X=x$, $X'$ is distributed according to the probabilities assigned by the retrodictive state $\gamma_x$ to the same POVM. Accordingly, $v(X)$ is the predictive feature value and $v(X')$ its retrodictive counterpart. Since both marginals are the same,
\begin{align}
    v^{\mathsf T}\tN v
    &=
    \mathbb E_X[v(X)^2]
    -
    \mathbb E_{XX'}[v(X)v(X')] \notag\\
    &=
    \frac12\,\mathbb E_{XX'}\!\left[(v(X)-v(X'))^2\right].
    \label{eq:feature-error}
\end{align}
Thus, $\tN$ is positive semidefinite, and $v^{\mathsf T}\tN v$ is half the mean-square disagreement between the predictive and retrodictive values of the feature. This is the feature-level counterpart of the full-outcome criterion: the quadratic form vanishes precisely when the retrodictive state predicts the observed value of the feature with certainty.

To compare features independently of arbitrary biases, scales, and physical units, we restrict to zero-mean, unit-variance features:
\begin{equation*}
    v^{\mathsf T}De=0,
    \qquad
    v^{\mathsf T}Dv=1,
\end{equation*}
where $e:=(1,\ldots,1)^{\mathsf T}$.

\begin{definition}[Scalar feature non-objectivity]
\label{def:feature_non-objectivity}
The scalar feature non-objectivity of $\mathsf P$ relative to $\gamma$ is
\begin{equation}
    \nu_\gamma(\mathsf P)
    :=
    \min_{\substack{v^{\mathsf T}De=0\\v^{\mathsf T}Dv=1}}
    v^{\mathsf T}\tN v.
\end{equation}
\end{definition}

The Hirschfeld--Gebelein--R\'enyi maximum correlation between the predictive and retrodictive outcome labels $(X,X')$ is~\cite{Hirschfeld_1935,gebelein1941statistische,renyi1959measures}
\begin{equation}
    \xi_{\mathrm{HGR}}(X,X')
    :=
    \max_{\substack{v^{\mathsf T}De=w^{\mathsf T}De=0\\v^{\mathsf T}Dv=w^{\mathsf T}Dw=1}}
    v^{\mathsf T}Tw.
\end{equation}
Since the two marginals coincide and $D^{-1/2}TD^{-1/2}$ is symmetric and positive semidefinite, this maximum is attained with the same feature on both sides. Moreover, for every centered, unit-variance feature,
\begin{equation*}
    v^{\mathsf T}\tN v=1-v^{\mathsf T}Tv.
\end{equation*}
Consequently,
\begin{equation}\label{eq:HGR}
    \nu_\gamma(\mathsf P)
    =
    1-\xi_{\mathrm{HGR}}(X,X').
\end{equation}

Within the fixed nontrivial MPPP block, this quantity is strictly positive. Indeed, if $\nu_\gamma(\mathsf P)=0$, Eq.~\eqref{eq:feature-error} implies that an optimizing standardized feature satisfies $v(X)=v(X')$ almost surely. Since $v$ has unit variance, it is nonconstant. Grouping outcomes that have the same value of $v$ therefore gives a nontrivial deterministic coarse-graining for which the predictive and retrodictive labels agree with certainty. This would yield a nontrivial objective refinement of the MPPP, contradicting its maximality. Therefore,
\begin{equation}\label{eq:feature-strict}
    \nu_\gamma(\mathsf P)>0\;.
\end{equation}
Note that the definition of $\nu_\gamma(\mathsf P)$ itself applies to any POVM with at least two non-null outcomes; the restriction to a nontrivial MPPP block is used only to guarantee strict positivity.

Finally, applying Eq.~\eqref{eq:feature-error} to the indicator features $v^{(x')}_x:=\delta_{x,x'}$, or equivalently to their centered and normalized versions, shows that agreement of all indicator features is equivalent to the full-outcome criterion in Definition~\ref{def:measurement_objectivity}. Thus, objectivity of the complete outcome label is a special case of the feature analysis.

\ 

\noindent {\em The non-objectivity tensors.}---The scalar $\nu_\gamma(\mathsf P)$ retains only the smallest retrodictive discrepancy among standardized features. Equation~\eqref{eq:feature-error} shows that its underlying matrix $\tN=D-T$ quantifies half the mean-square disagreement between predictive and retrodictive feature values in every feature direction. This matrix is defined for any finite POVM and faithful prior. We now return from the fixed MPPP block to an arbitrary pair $(\mathsf P,\gamma)$; from this point on, $p_x$, $T$, $D$, and the tensors refer to this unrestricted pair. We call $\tN$ the total non-objectivity tensor.

We decompose this tensor as
\begin{equation}
    \tN=\tU+\tA,
\end{equation}
where $\tU$ is the unsharpness tensor and $\tA$ is the asymmetry tensor. Unsharpness quantifies the failure of the outcome effects to define mutually exclusive sharp properties, while asymmetry quantifies their noncommutativity with the prior. The latter has the interpretation of prior coherence between outcome sectors once the former vanishes. Explicitly,
\begin{subequations}
\begin{align}
    \tU_{xx'}&:=\delta_{xx'}\Tr[\gamma P_x]
    -\frac12\Tr[\gamma\{P_x,P_{x'}\}],\\
    \tA_{xx'}&:=\frac12\Tr[\gamma\{P_x,P_{x'}\}]
    -\Tr[\sqrt\gamma P_x\sqrt\gamma P_{x'}],
\end{align}
\end{subequations}
where $\{A,B\}:=AB+BA$ denotes the anticommutator. Note that the intermediate matrix
\begin{equation*}
    S_{xx'}:=\frac12\Tr[\gamma\{P_x,P_{x'}\}]
\end{equation*}
also has an estimation-theoretic meaning. Hall's prior-dependent least-squares construction~\cite{hall2004prior} identifies $S_{xx'}/p_x$ as the optimal estimate of the effect $P_{x'}$ from outcome $x$.

\begin{theorem}[Positive-semidefinite decomposition]
\label{theo:positive_semidefinite}
For every finite POVM $\mathsf P=\{P_x\}_{x\in\X}$ and every faithful prior $\gamma>0$, we have $\tU\ge0$, $\tA\ge0$, and, consequently, $\tN=\tU+\tA\ge0$.
\end{theorem}

\begin{proof}
Fix $c\in\mathbb C^\X$ and let $Q:=\sum_xc_xP_x$. Then
\begin{equation*}
c^\dagger\tU c
=\Tr\left[\gamma\left(\sum_x|c_x|^2P_x-\frac12(Q^\dagger Q+QQ^\dagger)\right)\right]\geq0,
\end{equation*}
because
\begin{equation*}
\sum_x|c_x|^2P_x-Q^\dagger Q
=\sum_x(c_x\one-Q)^\dagger P_x(c_x\one-Q)\geq0,
\end{equation*}
and, similarly,
\begin{equation*}
\sum_x|c_x|^2P_x-QQ^\dagger
=\sum_x(c_x\one-Q)P_x(c_x\one-Q)^\dagger\geq0.
\end{equation*}
Hence $\tU\geq0$.

Moreover,
\begin{equation*}
\tA_{xx'}
=\frac12\Tr\left([\sqrt\gamma,P_x]^\dagger[\sqrt\gamma,P_{x'}]\right).
\end{equation*}
Thus $\tA$ is a Gram matrix and $\tA\geq0$. Finally, $\tN=\tU+\tA\geq0$.
\end{proof}

The decomposition $\tN=\tU+\tA$ has a direct operator interpretation. For every real-valued feature $v=(v_x)_{x\in\X}$, recall the operator $O_v$ defined in Eq.~\eqref{eq:feature-observable} and define the second-moment operator
\begin{equation}\label{eq:feature_operator}
    O_{v^2}:=\sum_xv_x^2P_x.
\end{equation}
The three tensors then satisfy
\begin{subequations}
\begin{align}
    v^{\mathsf T}\tU v
    &=\Tr[\gamma O_{v^2}]-\Tr[\gamma O_v^2],\label{eq:first-identity}\\
    v^{\mathsf T}\tA v
    &=\Tr[\gamma O_v^2]-\Tr[\sqrt\gamma O_v\sqrt\gamma O_v]
      =\IWY(\gamma,O_v),\\
    v^{\mathsf T}\tN v
    &=\Tr[\gamma O_{v^2}]-\Tr[\sqrt\gamma O_v\sqrt\gamma O_v].
\end{align}
\end{subequations}
Here $\IWY(\gamma,O_v)$ is the Wigner--Yanase skew information~\cite{wigner1963information}. Thus, $\tU$ measures the mismatch between the classical second moment $O_{v^2}$ and the squared feature observable $O_v^2$, while $\tA$ measures the noncommutativity of $O_v$ with the prior. The following result shows that both contributions are faithful.

Note that the first identity, Eq.~\eqref{eq:first-identity}, connects $\tU$ with earlier variance-based notions of POVM unsharpness. The positive operator $O_{v^2}-O_v^2$ is usually called the intrinsic noise or uncertainty operator of the POVM with numerical outcomes $v_x$~\cite{hall2003algebra,hall2004prior,busch2004noise,massar2007uncertainty,busch2014rms}. More specifically, $\tU$ coincides entrywise with the state-dependent matrix introduced by Liu and Luo~\cite{liu2021quantifying}. In the present framework, this becomes one term of the retrodictive decomposition $\tN=\tU+\tA$.

\begin{theorem}[Zero conditions]
\label{thm:tensor-zero-conditions}
For every finite POVM $\mathsf P=\{P_x\}$ and faithful prior $\gamma>0$,
\begin{align}
    \tU=0
    &\quad\Longleftrightarrow\quad
    \mathsf P\text{ is a PVM},
    \label{eq:U-zero-condition}\\
    \tA=0
    &\quad\Longleftrightarrow\quad
    [\gamma,P_x]=0\quad\forall x.
    \label{eq:A-zero-condition}
\end{align}
Consequently, $\tN=\tU+\tA$ vanishes if and only if $\mathsf P$ is a PVM commuting with the prior. By Eq.~\eqref{eq:objective-sector-decomposition}, this is exactly the condition under which every outcome feature admits a pre-existing-value interpretation.
\end{theorem}

\begin{proof}
Suppose first that $\tU=0$. For the indicator feature $v^{(x)}_{x'}:=\delta_{xx'}$, one has $O_v=O_{v^2}=P_x$, and hence
\begin{equation*}
0=v^{\mathsf T}\tU v=\Tr[\gamma(P_x-P_x^2)].
\end{equation*}
Since $P_x-P_x^2\geq0$ and $\gamma>0$, this implies $P_x=P_x^2$ for every $x$. Thus $\mathsf P$ is a PVM. The converse follows immediately from $O_{v^2}=O_v^2$ for every feature $v$.

Similarly, $\tA=0$ implies, for every $x$,
\begin{equation*}
0=\frac12\Tr\left([\sqrt\gamma,P_x]^\dagger[\sqrt\gamma,P_x]\right),
\end{equation*}
and therefore $[\gamma,P_x]=0$. Conversely, if every $P_x$ commutes with $\gamma$, then $[\sqrt\gamma,O_v]=0$ for every $v$, and hence $\tA=0$.
\end{proof}

The two zero conditions identify successive obstructions to objectivity relative to $\gamma$. If $\tU\neq0$, the outcomes do not define mutually exclusive sharp properties in the first place. If $\tU=0$, the POVM is a PVM and its outcomes define sharp sectors; within this sharp regime, $\tA\neq0$ means that the prior contains coherence between those sectors and therefore does not admit the sector decomposition in Eq.~\eqref{eq:objective-sector-decomposition}. Their simultaneous vanishing is exactly the condition under which the outcome can be regarded as a faithful record of a sharp property present before the measurement.

\ 

\noindent {\em Monotonicity.}---We now show that, under a suitable class of prior-preserving quantum preprocessings, $\tU$ and $\tN$ cannot decrease, whereas $\tA$ cannot increase. Let
$\Phi:\mathcal B(\mathcal H)\to\mathcal B(\mathcal H)$ be a unital completely positive map acting on effects, and define
\begin{equation*}
    \Phi(\mathsf P):=\{\Phi(P_x)\}_x.
\end{equation*}
We consider maps that preserve both the prior and its modular structure.

\begin{definition}[$\gamma$-covariant maps]
\label{def:gamma-symmetric_map}
Let $\Phi^*$ denote the trace dual of $\Phi$, defined by
$\Tr[\Phi(A)B]=\Tr[A\Phi^*(B)]$. We call $\Phi$ $\gamma$-preserving if
\begin{equation}
    \Phi^*(\gamma)=\gamma.
\end{equation}
A $\gamma$-preserving map is $\gamma$-covariant if it is also covariant under the modular group of $\gamma$, namely,
\begin{equation}
    \Phi(\gamma^{it}X\gamma^{-it})
    =
    \gamma^{it}\Phi(X)\gamma^{-it},
    \qquad t\in\mathbb R.
\end{equation}
\end{definition}
Note that, in finite dimension, analytic continuation extends the above identity to
\begin{equation}\label{eq:complex-modular-covariance}
    \Phi(\gamma^sX\gamma^{-s})
    =
    \gamma^s\Phi(X)\gamma^{-s},
    \qquad s\in\mathbb C.
\end{equation}
In what follows, the cases $s=1/2$ and $s=1/4$ will be particularly useful.

Unital complete positivity ensures that $\Phi(\mathsf P)$ remains a POVM, while $\gamma$ preservation gives
\begin{equation*}
    \Tr[\gamma\Phi(P_x)]
    =
    \Tr[\Phi^*(\gamma)P_x]
    =
    \Tr[\gamma P_x].
\end{equation*}
Thus, $\mathsf P$ and $\Phi(\mathsf P)$ have the same reference outcome distribution and the same space of standardized features. Modular covariance also ensures that $[\gamma,X]=0$ implies $[\gamma,\Phi(X)]=0$. The asymmetry inequality below is closely related to known monotonicity properties of skew information~\cite{marvian2016how,takagi2019skew}; here it forms one part of the coupled tensor statement for unsharpness, asymmetry, and total non-objectivity.

\begin{theorem}[Monotonicity of the non-objectivity tensors]
\label{theo:monotonicity}
Let $\Phi$ be a $\gamma$-covariant map, and let $\tU'$, $\tA'$, and $\tN'$ denote the tensors associated with the transformed POVM $\Phi(\mathsf P)$. Then
\begin{equation}
    \tU'\geq\tU,
    \qquad
    \tA'\leq\tA,
    \qquad
    \tN'\geq\tN.
\end{equation}
\end{theorem}

\begin{proof}
Fix a real feature $v$ and write
\begin{equation*}
    O_v'=\Phi(O_v),
    \qquad
    O_{v^2}'=\Phi(O_{v^2}).
\end{equation*}
The first inequality requires only $\gamma$ preservation. Indeed, the Kadison--Schwarz inequality~\cite{bhatia2009positive} gives
\begin{align*}
v^{\mathsf T}\tU'v
&=\Tr[\gamma\Phi(O_{v^2})]-\Tr[\gamma\Phi(O_v)^2]\nonumber\\
&\geq\Tr[\gamma\Phi(O_{v^2}-O_v^2)]
=v^{\mathsf T}\tU v,
\end{align*}
where the last equality follows from $\Phi^*(\gamma)=\gamma$.

For the remaining inequalities, introduce
\begin{equation*}
    \Delta_\gamma^s(X):=\gamma^sX\gamma^{-s},
    \qquad
    \|X\|_\gamma^2:=\Tr[\gamma X^\dagger X].
\end{equation*}
$\gamma$-preservation and Kadison--Schwarz imply the GNS contraction
\begin{equation*}
    \|\Phi(X)\|_\gamma^2
    \leq\Tr[\gamma\Phi(X^\dagger X)]
    =\|X\|_\gamma^2.
\end{equation*}
Moreover, analytically continued modular covariance gives $\Delta_\gamma^s\circ\Phi=\Phi\circ\Delta_\gamma^s$, for $s\in\mathbb C$.
We use this relation below with $s=1/2$ for $\tA$ and with $s=1/4$ for $\tN$.

Indeed, for every self-adjoint $H$,
\begin{equation*}
    \Tr[\gamma H^2]-\Tr[\sqrt\gamma H\sqrt\gamma H]
    =\frac12\|(\Delta_\gamma^{1/2}-\mathrm{id})(H)\|_\gamma^2.
\end{equation*}
Applying this identity to $H=O_v'=\Phi(O_v)$ and using modular covariance Eq.~\eqref{eq:complex-modular-covariance} with $s=1/2$, we obtain
\begin{equation*}
    (\Delta_\gamma^{1/2}-\mathrm{id})(\Phi(O_v))
    =
    \Phi\bigl((\Delta_\gamma^{1/2}-\mathrm{id})(O_v)\bigr).
\end{equation*}
The GNS contraction therefore gives
\begin{align*}
v^{\mathsf T}\tA'v
&=\frac12\left\|\Phi\left((\Delta_\gamma^{1/2}-\mathrm{id})(O_v)\right)\right\|_\gamma^2\nonumber\\
&\leq\frac12\|(\Delta_\gamma^{1/2}-\mathrm{id})(O_v)\|_\gamma^2
=v^{\mathsf T}\tA v.
\end{align*}

Finally, the identity
\begin{equation*}
    \Tr[\sqrt\gamma H\sqrt\gamma H]
    =\|\Delta_\gamma^{1/4}(H)\|_\gamma^2
\end{equation*}
treats the retrodictive term in $\tN$. Using Eq.~\eqref{eq:complex-modular-covariance} with $s=1/4$ gives
\begin{equation*}
    \Delta_\gamma^{1/4}(\Phi(O_v))
    =
    \Phi(\Delta_\gamma^{1/4}(O_v)).
\end{equation*}
$\gamma$-preservation leaves $\Tr[\gamma O_{v^2}]$ unchanged, while the GNS contraction then yields
\begin{align*}
v^{\mathsf T}\tN'v
&=\Tr[\gamma O_{v^2}]
-\|\Phi(\Delta_\gamma^{1/4}(O_v))\|_\gamma^2\nonumber\\
&\geq\Tr[\gamma O_{v^2}]
-\|\Delta_\gamma^{1/4}(O_v)\|_\gamma^2
=v^{\mathsf T}\tN v.
\end{align*}
Since the three inequalities hold for every real $v$, the claims follow.
\end{proof}

Thus, a $\gamma$-covariant preprocessing can remove prior asymmetry only by adding at least as much unsharpness:
\begin{equation}
    0\leq\tA-\tA'\leq\tU'-\tU.
\end{equation}
It may suppress coherence relative to the prior, but it cannot thereby create a pre-existing sharp property; total non-objectivity cannot decrease. The same tensor inequality immediately yields the corresponding scalar monotonicity.

\begin{corollary}[Monotonicity of scalar feature non-objectivity]
Let $\Phi$ be a $\gamma$-covariant map. Then
\begin{equation}
    \nu_\gamma\!\left(\Phi(\mathsf P)\right)
    \geq
    \nu_\gamma(\mathsf P).
\end{equation}
Equivalently, the HGR maximum correlation between the corresponding predictive and retrodictive labels cannot increase under $\Phi$.
\end{corollary}

\begin{proof}
$\gamma$-preservation leaves the reference outcome distribution, and hence the set of standardized features, unchanged. Taking the minimum of the inequality $v^{\mathsf T}\tN'v\geq v^{\mathsf T}\tN v$ over this common set proves the first claim. The HGR statement follows from $\nu_\gamma(\mathsf P)=1-\xi_{\mathrm{HGR}}(X,X')$.
\end{proof}

\ 

\noindent {\em Conclusion.}---We have developed a retrodictive quantification of when a measurement is objective relative to $\gamma$. Its zero condition has a direct property-revealing meaning: when the POVM is sharp and commutes with the prior, the prior admits a canonical mixture decomposition over mutually exclusive outcome sectors, and the measurement can be regarded as revealing which pre-existing sharp property obtains. This interpretation applies to every real-valued function of the outcome. It is a prior-relative statement about the observable--prior pair: it neither fixes the physical preparation history of $\gamma$ nor depends on a particular instrument implementing the observable.

For a general POVM, the MPPP first extracts the finest sharp sector label revealed by the outcome. Within each remaining nontrivial MPPP block, three positive-semidefinite tensors resolve non-objectivity feature by feature. The total tensor measures disagreement between predictive and retrodictive feature values and decomposes exactly into unsharpness, which obstructs sharp properties, and Wigner--Yanase asymmetry, which obstructs prior compatibility. Optimizing the total tensor yields the HGR feature non-objectivity.

The tensor formulation also shows how these obstructions change under physical processing. A $\gamma$-covariant preprocessing cannot increase asymmetry, but it creates at least as much unsharpness as the asymmetry it removes, so total non-objectivity cannot decrease. Removing coherence relative to the prior is therefore not enough to make a measurement objective relative to $\gamma$: sharpness must be preserved as well. More broadly, our framework opens a new direction in which the pre-existence of measured properties is treated as a quantitative, feature-dependent notion rather than merely a yes-or-no condition.

\ 

\textit{Acknowledgments.}---F.~B. thanks M.~Hamed Mohammady for insightful discussions. J.~K. and F.~B. acknowledge support from MEXT Quantum Leap Flagship Program (MEXT QLEAP) Grant No.~ JPMXS0120319794. T.~N. acknowledges support from JST ERATO Grant Number JPMJER2402. K.~T. acknowledges support from JST SPRING Grant No.~JPMJSP2125. F.~B. acknowledges support also from JSPS KAKENHI Grants No.~23K03230 and~26K00621.  

\bibliography{library}

\end{document}